\RequirePackage{ifthen}
\RequirePackage{ifpdf}
\newcommand{\driverOption}{}
\ifthenelse{\boolean{pdf}}{
  \renewcommand{\driverOption}{pdftex}
} { 
  \renewcommand{\driverOption}{dvips}
}

\documentclass[reqno, 11pt, letterpaper, commented, oneside, \driverOption]{amsart}

\newboolean{isCommented}
\usepackage{geometry}
\usepackage{bbm}
\usepackage[final]{graphicx}
\usepackage{upref}
\usepackage{enumitem}
\usepackage{latexsym}
\usepackage{amssymb}
\usepackage[ansinew]{inputenc}
\usepackage[english]{babel}
\usepackage[T1]{fontenc}
\usepackage[ruled,vlined,norelsize]{algorithm2e}  
\usepackage{mathtools}
\usepackage{color}
\usepackage{xcolor}
\usepackage{units}
\usepackage{mycommands}

\ifthenelse{\boolean{pdf}}{
  \usepackage[final]{ps4pdf}
} { 
  \usepackage[inactive]{ps4pdf}
}
\PSforPDF{\usepackage{psfrag}}

\newcommand{\hyperrefDriverOption}{}
\ifthenelse{\boolean{pdf}}{
	\renewcommand{\hyperrefDriverOption}{pdftex}
} { 
	\renewcommand{\hyperrefDriverOption}{hypertex}
}
\usepackage[\hyperrefDriverOption,
  colorlinks = false, 
  pdftitle={A constant-time algorithm for middle levels Gray codes}]
  {hyperref}
\hypersetup{pdfauthor={Torsten M\374tze, Jerri Nummenpalo}}

\ifthenelse{\boolean{isCommented}} {
	\newcommand{\TM}[1]{\marginpar{\parbox{4cm}{{\small {\bf TM:} #1}}}}
	\newcommand{\JN}[1]{\marginpar{\parbox{4cm}{{\small {\bf JN:} #1}}}}
} { 
	\newcommand{\TM}[1]{}
	\newcommand{\JN}[1]{}
}

\newtheorem{theorem}{Theorem}
\newtheorem{lemma}[theorem]{Lemma}

\theoremstyle{definition}

\theoremstyle{remark}

\ifthenelse{\boolean{isCommented}} {
	\geometry{
	  hmargin={25mm, 50mm},
	  marginparwidth=40mm, 
	  vmargin={25mm, 25mm},
	  headsep=10mm,
	  headheight=5mm,
	  footskip=10mm
	}
} { 
	\geometry{
	  hmargin={25mm, 25mm}, 
	  vmargin={25mm, 25mm},
	  headsep=10mm,
	  headheight=5mm,
	  footskip=10mm
	}
}

\begin{document}

\begin{center}

\renewcommand{\thefootnote}{\fnsymbol{footnote}}

\LARGE A constant-time algorithm for middle levels Gray codes\footnote{An extended abstract of this paper appeared in the Proceedings of the 28th Annual ACM-SIAM Symposium on Discrete Algorithms, {SODA} 2017 \cite{DBLP:conf/soda/MutzeN17}.
Torsten M\"utze is also affiliated with Charles University, Faculty of Mathematics and Physics, and was supported by Czech Science Foundation grant GA~19-08554S, and by German Science Foundation grant~413902284.}

\vspace{2mm}

\small

\begingroup
\begin{tabular}{l@{\hspace{2em}}l}
  \Large Torsten M\"utze & \Large Jerri Nummenpalo \\[2mm]
  Institut f\"ur Mathematik & Department of Computer Science \\
  TU Berlin & ETH Z\"urich \\
  10623 Berlin, Germany & 8092 Z\"urich, Switzerland \vspace{.05mm} \\
  {\small {\tt muetze@math.tu-berlin.de}} & {\small {\tt njerri@inf.ethz.ch}}
\end{tabular}%
\endgroup

\vspace{5mm}

\small

\begin{minipage}{0.8\linewidth}
\textsc{Abstract.}
For any integer~$n\geq 1$, a \emph{middle levels Gray code} is a cyclic listing of all $n$-element and $(n+1)$-element subsets of $\{1,2,\ldots,2n+1\}$ such that any two consecutive sets differ in adding or removing a single element.
The question whether such a Gray code exists for any~$n\geq 1$ has been the subject of intensive research during the last 30 years, and has been answered affirmatively only recently [T.~M\"utze. Proof of the middle levels conjecture. \textit{Proc. London Math. Soc.}, 112(4):677--713, 2016].
In a follow-up paper [T.~M\"utze and J.~Nummenpalo. An efficient algorithm for computing a middle levels Gray code. \textit{ACM Trans. Algorithms}, 14(2):29~pp., 2018] this existence proof was turned into an algorithm that computes each new set in the Gray code in time~$\cO(n)$ on average.
In this work we present an algorithm for computing a middle levels Gray code in optimal time and space: each new set is generated in time~$\cO(1)$ on average, and the required space is~$\cO(n)$.
\end{minipage}

\vspace{2mm}

\begin{minipage}{0.8\linewidth}
\end{minipage}

\vspace{2mm}

\end{center}

\vspace{5mm}

\section{Introduction}

Efficiently generating all objects in a particular combinatorial class such as permutations, subsets, partitions, trees, strings etc.\ is one of the oldest and most fundamental algorithmic problems.
Such generation algorithms are used as building blocks in a wide range of practical applications; the survey~\cite{MR1491049} lists numerous references.
In fact, more than half of the most recent volume of Knuth's seminal series \emph{The Art of Computer Programming}~\cite{MR3444818} is devoted to this fundamental subject.
The ultimate goal for these problems is to come up with algorithms that generate each new object in constant time, entailing that consecutive objects may differ only in a constant amount.
For such an algorithm, `generating an object' means constructing a suitable representation of the object in memory.
In an actual application, each such construction step would be followed by a call to a function that utilizes the object for some user-defined purpose, such as computing the value of an objective function to be optimized.
After an object is constructed in memory, the memory can be reused and modified for storing the next object.
`Constant time per object' means that the total time (=arithmetic complexity) spent by the algorithm for generating all objects, divided by the number of objects generated, is a constant.
Typically, the number of objects is exponential in some parameter (e.g., the number of permutations of $n$ objects is $n!=2^{\Theta(n \log n)}$), and so this quotient should not depend on the parameter.
Such constant-time generation algorithms are known for several combinatorial classes, and many of these results are covered in the classical books~\cite{MR0396274,MR993775}.
To mention some concrete examples, constant-time algorithms are known for the following problems:
\begin{enumerate}[label=(\arabic*),topsep=0mm,leftmargin=7mm]
\item generating all permutations of $[n]:=\{1,2,\ldots,n\}$ by adjacent transpositions \cite{MR0159764,Trotter:1962,MR0502206,MR0464682},
\item generating all subsets of~$[n]$ by adding or removing an element in each step \cite{gray:patent},
\item generating all $k$-element subsets of~$[n]$ by exchanging an element in each step \cite{MR0366085,MR0424386,MR821383,MR782221,MR936104},
\item generating all binary trees with $n$~vertices by rotation operations \cite{MR796208,MR920505,MR1239499},
\item generating all spanning trees of a graph by exchanging an edge in each step \cite{MR0245357,MR0246789,MR0307952}.
\end{enumerate}

In this paper we revisit the well-known problem of generating all $n$-element and $(n+1)$-element subsets of~$[2n+1]$ by adding or removing a single element in each step.
In a computer these subsets are naturally represented by bitstrings of length~$2n+1$, with 1-bits at the positions of the elements contained in the set and 0-bits at the remaining positions.
Consequently, the problem is equivalent to generating all bitstrings of length~$2n+1$ with weight~$n$ or~$n+1$, where the \emph{weight} of a bitstring is the number of 1s in it.
We refer to such a Gray code as as \emph{middle levels Gray code}.
Clearly, a middle levels Gray code has $N:=\binom{2n+1}{n}+\binom{2n+1}{n+1}=2^{\Theta(n)}$ many bitstrings in total, and the weight alternates between~$n$ and~$n+1$ in each step.
The existence of a middle levels Gray code for any~$n\geq 1$ is asserted by the well-known \emph{middle levels conjecture}, raised independently in the 80s by Havel~\cite{MR737021} and Buck and Wiedemann~\cite{MR737262}.
The conjecture has also been attributed to Dejter, Erd{\H{o}}s, Trotter~\cite{MR962224} and various others, it appears in the popular books~\cite{MR2034896,MR3444818,MR2858033}, and it is mentioned in Gowers' recent expository survey on Peter Keevash's work~\cite{MR3584100}.
The middle levels conjecture has attracted considerable attention over the last 30~years~\cite{savage:93,MR1350586,MR1329390,MR2046083,MR962223,MR962224,MR1268348,MR2195731,MR2609124,MR2946389,MR2548541,shimada-amano}, and a positive solution, i.e., an existence proof for a middle levels Gray code for any~$n\geq 1$, has been announced only recently.

\begin{theorem}[\cite{MR3483129,mlc-short:18}]
\label{thm:middle-levels}
A~middle levels Gray code exists for any~$n\geq 1$.
\end{theorem}

In a follow-up paper~\cite{muetze-nummenpalo:18}, this existence argument was turned into an algorithm for computing a middle levels Gray code.

\begin{theorem}[\cite{muetze-nummenpalo:18}]
\label{thm:algo1}
There is an algorithm, which for a given bitstring of length~$2n+1$, $n\geq 1$, with weight~$n$ or~$n+1$ computes the next~$\ell\geq 1$ bitstrings in a middle levels Gray code in time~$\cO(\ell n+n^2)$.
\end{theorem}

Clearly, the running time of this algorithm is~$\cO(n)$ on average per generated bitstring for $\ell=\Omega(n)$.
However, this falls short of the optimal $\cO(1)$~time bound one could hope for, given that in each step only a single bit needs to be flipped, which is a constant amount of change.

\subsection{Our results}

In this paper we present an algorithm for computing a middle levels Gray code in optimal time and space.

\begin{theorem}
\label{thm:algo2}
There is an algorithm, which for a given bitstring of length~$2n+1$, $n\geq 1$, with weight~$n$ or~$n+1$ computes the next~$\ell\geq 1$ bitstrings in a middle levels Gray code in time~$\cO(\ell+n)$.
\end{theorem}

Clearly, the running time of this algorithm is~$\cO(1)$ on average per generated bitstring for $\ell=\Omega(n)$, and the required initialization time~$\cO(n)$ and the required space~$\cO(n)$ are also optimal.

We implemented our new middle levels Gray code in C++, and we invite the reader to experiment with this code, which can be found and run on the Combinatorial Object Server website~\cite{cos}.
As a benchmark, we used this code to compute a middle levels Gray code for~$n=19$ in 20~minutes on a standard desktop computer.
This is by a factor of~$72$ faster than the 24~hours reported in~\cite{muetze-nummenpalo:18} for the algorithm from Theorem~\ref{thm:algo1}, and by four orders of magnitude faster than the 164~days previously needed for a brute-force search~\cite{shimada-amano}.
Note that a middle levels Gray code for~$n=19$ consists of $N=137.846.528.820\approx 10^{11}$ bitstrings.
For comparison, a program that only consists of a loop with a counting variable running from $1,\ldots,N$ and nothing else was only by a factor of~$5$ faster (4~minutes) than our middle levels Gray code computation on the same hardware.
Roughly speaking, we need about 5~arithmetic operations for producing the next bitstring in the Gray code.

We now also obtain efficient algorithms for a number of related Gray codes that have been constructed using Theorem~\ref{thm:middle-levels} as an induction basis.
These Gray codes consist of several combined middle levels Gray codes of smaller dimensions.
Specifically, it was a long-standing problem (see \cite{MR1152123,MR1271867,MR1778200,MR1999733}) to construct a Gray code that lists all $k$-element and $(n-k)$-element subsets of~$[n]$, where $n\geq 2k+1$, by either adding or removing $n-2k$~elements in each step.
This was solved in~\cite{MR3759914}, and using Theorem~\ref{thm:algo2} this construction can now be turned into an efficient algorithm.
Moreover, in~\cite{MR3758308} Theorem~\ref{thm:algo2} is used to derive constant-time algorithms for generating minimum-change listings of all $n$-bit strings whose weight is in some interval~$[k,l]$, $0\leq k\leq l\leq n$, a far-ranging generalization of the middle levels conjecture and the classical problems~(2) and~(3) mentioned before.

In this work we restrict our attention to computing one particular `canonical' middle levels Gray code for any~$n\geq 1$, even though we know from~\cite{MR3483129} that there are double-exponentially many different ones (recall~\cite[Remark~3]{muetze-nummenpalo:18}).

\subsection{Making the algorithm loopless}

We shall see that most steps of our algorithm require only constant time~$\cO(1)$ in the worst case to generate the next bitstring, but after every sequence of~$\Theta(n)$ such `fast' steps, a `slow' step which requires time~$\Theta(n)$ is encountered, yielding constant average time performance.
Therefore, our algorithm could easily be transformed into a loopless algorithm, i.e., one with a~$\cO(1)$ \emph{worst case} bound for each generated bitstring, by introducing an additional FIFO queue of size~$\Theta(n)$ and by simulating the original algorithm such that during every sequence of~$d$ `fast' steps, $d-1$ results are stored in the queue and only one of them is returned, and during the `slow' steps the queue is emptied at the same speed.
For this the constant~$d$ must be chosen so that the queue is empty when the `slow' steps are finished.
This idea of delaying the output to achieve a loopless algorithm is also used in~\cite{DBLP:conf/fun/HerterR16} (see also~\cite[Section~1]{MR0464682}).
Even though the resulting algorithm would indeed be loopless, it would still be slower than the original algorithm, as it produces every bitstring only after it was produced in the original algorithm, due to the delay caused by the queue and the additional instructions for queue management.
In other words, the hidden constant in the $\cO(1)$~bound for the modified algorithm is higher than for the original algorithm, so this loopless algorithm is only of theoretical interest, and we will not discuss it any further.

\subsection{Ingredients}

Our algorithm for computing a middle levels Gray code implements the strategy of the short proof of Theorem~\ref{thm:middle-levels} presented in~\cite{mlc-short:18}.
In the most basic version, the algorithm computes several short cycles that together visit all bitstrings of length~$2n+1$ with weight~$n$ or~$n+1$.
We then modify a few steps of the algorithm so that these short cycles are joined to a Gray code that visits all bitstrings consecutively.

Let us briefly discuss the main differences between the algorithms from Theorem~\ref{thm:algo1} and Theorem~\ref{thm:algo2} and the improvements that save us a factor of~$n$ in the running time.
At the lowest level, the algorithm from Theorem~\ref{thm:algo1} consists of a rather unwieldy recursion, which for any given bitstring computes the next one in a middle levels Gray code.
This recursion runs in time~$\Theta(n)$, and therefore represents one of the bottlenecks in the running time.
In addition, there are various high-level functions that are called every $\Theta(n)$~many steps and run in time~$\Theta(n^2)$, and which therefore also represent $\Theta(n)$~bottlenecks.
These high-level functions control which subsets of bitstrings are visited in which order, to make sure that each bitstring is visited exactly once.

To overcome these bottlenecks, we replace the recursion at the lowest level by a simple combinatorial technique, first proposed in~\cite{MR3738156} and heavily used in the short proof of Theorem~\ref{thm:middle-levels} presented in~\cite{mlc-short:18}.
This technique allows us to compute for certain `special' bitstrings that are encountered every $\Theta(n)$~many steps, a sequence of bit positions to be flipped during the next $\Theta(n)$~many steps.
Computing such a flip sequence can be done in time~$\Theta(n)$, and when this is accomplished each subsequent step takes only constant time:
We simply flip the precomputed positions one after the other, until the next `special' bitstring is encountered and the flip sequence has to be recomputed.
The high-level functions in the new algorithm are very similar as in the old one.
We cut down their running time by a factor of~$n$ (from quadratic to linear) by using more sophisticated data structures and by resorting to well-known algorithms such as Booth's linear-time algorithm~\cite{MR585391} for computing the lexicographically smallest rotation of a given string.

\subsection{Outline of this paper}

In Section~\ref{sec:basic-defs} we introduce important definitions that will be used throughout the paper.
In Section~\ref{sec:algorithm} we present our new middle levels Gray code algorithm.
In Section~\ref{sec:correctness} we prove the correctness of the algorithm, and in Section~\ref{sec:running-time} we discuss how to implement it to achieve the claimed runtime and space bounds.

\section{Preliminaries}
\label{sec:basic-defs}

\textit{Operations on sequences and bitstrings.}
We let $(a_1,\ldots,a_n)$ denote the sequence of integers $a_1,\ldots,a_n$.
We generalize this notation allowing~$a_i$ to be itself an integer sequence: In that case, if $a_i=(b_1,\ldots,b_m)$, then $(a_1,\ldots,a_n)$ is shorthand for $(a_1,\ldots,a_{i-1},b_1,\ldots,b_m,a_{i+1},\ldots,a_n)$.
The empty integer sequence is denoted by~$()$.
For any sequence~$a$, we let~$|a|$ denote its length.
For any integer~$k\geq 0$ and any bitstring~$x$, we write~$x^k$ for the concatenation of $k$~copies of~$x$.
Moreover, $\rev(x)$ denotes the reversed bitstring, and~$\ol{x}$ denotes the bitstring obtained by taking the complement of every bit in~$x$.
We also define $\ol{\rev}(x):=\rev(\ol{x})=\ol{\rev(x)}$.
For any graph~$G$ whose vertices are bitstrings and any bitstring~$x$, we write~$Gx$ for the graph obtained from~$G$ by appending~$x$ to all vertices.

\textit{Bitstrings and lattice paths.}
We let~$B_{n,k}$ denote the set of all bitstrings of length~$n$ with weight~$k$.
Any bitstring $x\in B_{n,k}$ can be interpreted as a lattice path as follows; see Figure~\ref{fig:bij}:
We read~$x$ from left to right and draw a path in the integer lattice $\mathbb{Z}_2$ that starts at the origin~$(0,0)$.
For every 1-bit encountered in~$x$, we draw an \emph{$\upstep$-step} that changes the current coordinate by~$(+1,+1)$, and for every 0-bit encountered in~$x$, we draw a \emph{$\downstep$-step} that changes the current coordinate by~$(+1,-1)$.
Note that the resulting lattice path ends at the coordinate~$(n,2k-n)$.
We let~$D_n$ denote the bitstrings from~$B_{2n,n}$ with the property that in every prefix, there are at least as many 1s as 0s.
Moreover, we let~$D_n^-$ denote the bitstrings from~$B_{2n,n}$ that have this property for all but exactly one prefix.
In terms of lattice paths, $D_n$ are the paths with $2n$ steps that end at the abscissa~$y=0$ and that never move below this line, commonly known as \emph{Dyck paths}, whereas~$D_n^-$ are the paths that move below the abscissa~$y=0$ exactly once.
It is well-known that $|D_n|=|D_n^-|$ and that this quantity is given by the $n$th Catalan number.
We also define $D:=\bigcup_{n\geq 0}D_n$.
Any nonempty~$x\in D$ can be written uniquely as $x=1u0v$ with $u,v\in D$.
Similarly, any~$x\in D^-$ can be written uniquely as $x=u01v$ with~$u,v\in D$.
We refer to this as the \emph{canonical decomposition of~$x$}; see Table~\ref{tab:paths-q6}.

\begin{figure}
\centering
\PSforPDF{
 \psfrag{x}{\parbox{5cm}{a bitstring \\ $x=1101101000$}}
 \psfrag{p}{\parbox{3cm}{a lattice path \\ from~$D_{10}$}}
 \psfrag{t}{\parbox{4cm}{a rooted tree}}
 \psfrag{z}{0}
 \psfrag{ten}{10}
 \includegraphics{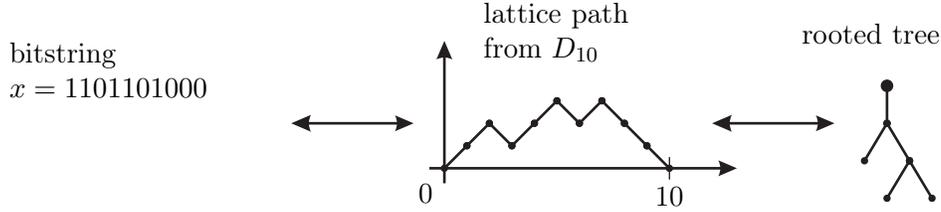}
}
\caption{Bijections between bitstrings and lattice paths (left), and between Dyck paths and rooted trees (right).}
\label{fig:bij}
\end{figure}

\textit{Rooted trees.}
An \emph{(ordered) rooted tree} is a tree with a specified root vertex, and the children of each vertex have a specified left-to-right ordering.
We think of a rooted tree as a tree embedded in the plane with the root on top, with downward edges leading from any vertex to its children, and the children appear in the specified left-to-right ordering.
Using a standard Catalan bijection, every Dyck path~$x\in D_n$ can be interpreted as a rooted tree with $n$~edges; see~\cite{MR3467982} and Figure~\ref{fig:bij}.
We therefore refer to the elements of~$D_n$ also as rooted trees.
Given a rooted tree~$x$, the rotation operation $\trot(x)$ shifts the root to the leftmost child of the root; see Figure~\ref{fig:rot}.
In terms of bitstrings, if $x=1u0v$ is the canonical decomposition of~$x$, then $\trot(x)=u1v0$.

\textit{The middle levels graph~$G_n$.}
We describe our algorithm to compute a middle levels Gray code using the language of graph theory.
We let~$G_n$ denote the \emph{middle levels graph}, which has all bitstrings of length~$2n+1$ with weight~$n$ or~$n+1$ as vertices, with an edge between any two bitstrings that differ in exactly one bit.
Clearly, computing a middle levels Gray code is equivalent to computing a Hamilton cycle in~$G_n$.
We let~$H_n$ denote the graph obtained by considering the subgraph of~$G_n$ induced by all vertices whose last bit equals~0, and by removing the last bit from every vertex.
Note that~$G_n$ consists of a copy of~$H_n\, 0$, a copy of $\ol{\rev}(H_n)\, 1$, plus the matching $M_n:=\{(x0,x1)\mid x\in B_{2n,n}\}$; see Figure~\ref{fig:2factor}.
The matching edges are the edges along which the last bit is flipped.

\section{The algorithm}
\label{sec:algorithm} 

Our algorithm to compute a Hamilton cycle in the middle levels graph~$G_n$ consists of several nested functions (see Algorithm~\ref{alg:hamcycle} below), and in the following we explain these functions from bottom to top.
The low-level functions compute paths in~$G_n$, and the high-level functions combine them to a Hamilton cycle.

\subsection{Computing paths in~$H_n$}
\label{sec:paths}

\begin{figure}
\centering
\PSforPDF{
 \psfrag{x}{$x\in D_n$}
 \psfrag{xb}{\footnotesize $x=111001110011110000001100$}
 \psfrag{sigmax}{\footnotesize $\sigma(x)=(20,1,5,2,4,3,2,4,1,5,19,6,10,7,9,8,7,9,6,10,18,11,17,12,16,13,15,14,13,15,12,16,11,17,10,18,5,19)$}
 \psfrag{n01}{\footnotesize $1$}
 \psfrag{n02}{\footnotesize $2$}
 \psfrag{n03}{\footnotesize $3$}
 \psfrag{n04}{\footnotesize $4$}
 \psfrag{n05}{\footnotesize $5$}
 \psfrag{n06}{\footnotesize $6$}
 \psfrag{n07}{\footnotesize $7$}
 \psfrag{n08}{\footnotesize $8$}
 \psfrag{n09}{\footnotesize $9$}
 \psfrag{n10}{\footnotesize $10$}
 \psfrag{na1}{\footnotesize $11$}
 \psfrag{n12}{\footnotesize $12$}
 \psfrag{n13}{\footnotesize $13$}
 \psfrag{n14}{\footnotesize $14$}
 \psfrag{n15}{\footnotesize $15$}
 \psfrag{n16}{\footnotesize $16$}
 \psfrag{n17}{\footnotesize $17$}
 \psfrag{n18}{\footnotesize $18$}
 \psfrag{n19}{\footnotesize $19$}
 \psfrag{n20}{\footnotesize $20$}
 \psfrag{n21}{\footnotesize $21$}
 \psfrag{n22}{\footnotesize $22$}
 \psfrag{n23}{\footnotesize $23$}
 \psfrag{n24}{\footnotesize $24$}
 \psfrag{n01r}{\color{red}\footnotesize $1$}
 \psfrag{n02r}{\color{red}\footnotesize $2$}
 \psfrag{n03r}{\color{red}\footnotesize $3$}
 \psfrag{n04r}{\color{red}\footnotesize $4$}
 \psfrag{n05r}{\color{red}\footnotesize $5$}
 \psfrag{n06r}{\color{red}\footnotesize $6$}
 \psfrag{n07b}{\color{blue}\footnotesize $7$}
 \psfrag{n08b}{\color{blue}\footnotesize $8$}
 \psfrag{n09b}{\color{blue}\footnotesize $9$}
 \psfrag{n10b}{\color{blue}\footnotesize $10$}
 \psfrag{nar}{\color{red}\footnotesize $11$}
 \psfrag{n12r}{\color{red}\footnotesize $12$}
 \psfrag{n13r}{\color{red}\footnotesize $13$}
 \psfrag{n14r}{\color{red}\footnotesize $14$}
 \psfrag{n15r}{\color{red}\footnotesize $15$}
 \psfrag{n16r}{\color{red}\footnotesize $16$}
 \psfrag{n17b}{\color{blue}\footnotesize $17$}
 \psfrag{n18b}{\color{blue}\footnotesize $18$}
 \psfrag{n19b}{\color{blue}\footnotesize $19$}
 \psfrag{n20b}{\color{blue}\footnotesize $20$}
 \psfrag{n21r}{\color{red}\footnotesize $21$}
 \psfrag{n22r}{\color{red}\footnotesize $22$}
 \psfrag{n23r}{\color{red}\footnotesize $23$}
 \psfrag{n24r}{\color{red}\footnotesize $24$}
 \psfrag{n25r}{\color{red}\footnotesize $25$}
 \psfrag{n26r}{\color{red}\footnotesize $26$}
 \psfrag{n27r}{\color{red}\footnotesize $27$}
 \psfrag{n28r}{\color{red}\footnotesize $28$}
 \psfrag{n29b}{\color{blue}\footnotesize $29$}
 \psfrag{n30b}{\color{blue}\footnotesize $30$}
 \psfrag{n31b}{\color{blue}\footnotesize $31$}
 \psfrag{n32b}{\color{blue}\footnotesize $32$}
 \psfrag{n33b}{\color{blue}\footnotesize $33$}
 \psfrag{n34b}{\color{blue}\footnotesize $34$}
 \psfrag{n35b}{\color{blue}\footnotesize $35$}
 \psfrag{n36b}{\color{blue}\footnotesize $36$}
 \psfrag{n37b}{\color{blue}\footnotesize $37$}
 \psfrag{n38b}{\color{blue}\footnotesize $38$}
 \includegraphics{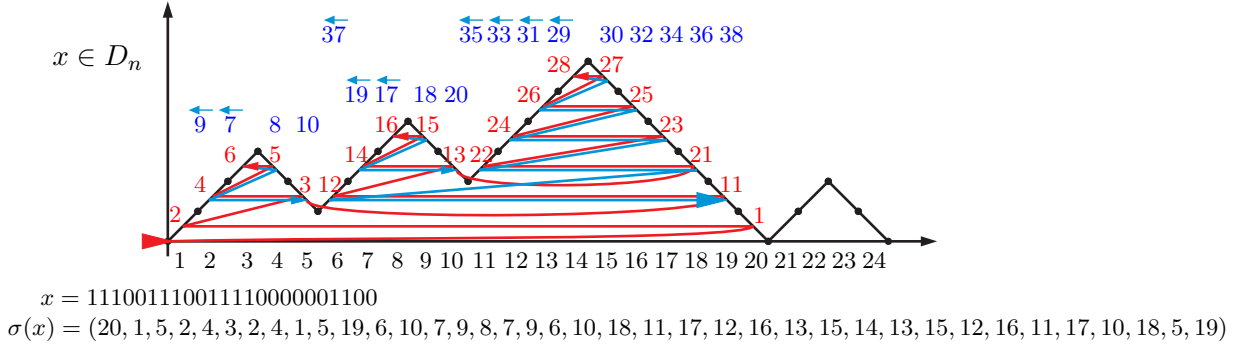}
}
\caption{Illustration of the recursive computation of the flip sequence~$\sigma(x)$.
The number $i=1,2,\ldots,38$ on the Dyck path indicates the position of~$x$ to be flipped in step~$i$, i.e., the $i$th entry of~$\sigma(x)$, where the little left-arrows act as modifiers that change the position drawn in the figure by~$-1$.
E.g., the 11th entry of~$\sigma(x)$ is the position of the $\downstep$-step of~$x$ marked with 11, so $\sigma(11)=19$; the 7th entry of $\sigma(x)$ is the position of the $\upstep$-step of~$x$ marked with 7 minus 1, so $\sigma(7)=3-1=2$.}
\label{fig:sigma}
\end{figure}

In this section we describe a set of disjoint paths~$\cP_n$ that together visit all vertices of the graph~$H_n$.
The starting vertices of these paths are the vertices $x\in D_n$, and in the following we describe a rule~$\sigma(x)$ that specifies the sequence of bit positions to be flipped along the path starting at~$x$.
To compute the flip sequence~$\sigma(x)$ for a given vertex $x\in D_n$, we interpret~$x$ as a Dyck path, and we alternatingly flip $\downstep$-steps and $\upstep$-steps of this Dyck path (corresponding to 0s and 1s in~$x$, respectively).
Specifically, for $x\in D_n$ we consider the canonical decomposition $x=1u0v$ and define $a:=1$, $b:=|u|+2$ and
\begin{subequations}
\label{eq:sigma-rec}
\begin{equation}
\label{eq:sigma}
  \sigma(x):=(b,\, a,\, \sigma_{a+1}(u)) \enspace,
\end{equation}
where~$\sigma_a(x')$ is defined for any substring $x'\in D$ of~$x$ starting at position~$a$ in~$x$ by considering the canonical decomposition $x'=1u'0v'$, by defining $b:=a+|u'|+1$ and by recursively computing
\begin{equation}
\label{eq:sigmap}
  \sigma_a(x') := \begin{cases}
      () & \text{if } |x'|=0 \enspace, \\
      \big(b,\, a,\, \sigma_{a+1}(u'),\, a-1,\, b,\, \sigma_{b+1}(v')\big) & \text{otherwise} \enspace.
      \end{cases}
\end{equation}
\end{subequations}
Note that in~\eqref{eq:sigma}, the integers $a$ and $b$ are the positions of~1 and~0 in the canonical decomposition~$x=1u0v$.
Similarly, in~\eqref{eq:sigmap}, the integers $a$ and $b$ are the positions of~1 and~0 in the canonical decomposition of the substring $x'=1u'0v'$ of~$x$, and those positions are with respect to the entire string~$x$ (the index of $\sigma_a$ tracks the starting position of~$x'$ in~$x$).

In words, the sequence~$\sigma(x)$ defined in \eqref{eq:sigma} first flips the $\downstep$-step immediately after the subpath~$u$ (at position~$b$ of~$x$), then the $\upstep$-step immediately before the subpath~$u$ (at position~$1$), and then recursively steps of~$u$.
No steps of~$v$ are flipped at all.
The sequence~$\sigma_a(x')$ defined in \eqref{eq:sigmap} first flips the $\downstep$-step immediately after the subpath~$u'$ (at position~$b$ of~$x$), then the $\upstep$-step immediately before the subpath~$u'$ (at position~$a$), then recursively steps of~$u'$, then again the step immediately to the left of~$x'$ (which is not part of~$x'$, hence the index~$a-1$), then again the step immediately to the right of~$u'$ (at position~$b$), and finally recursively steps of~$v'$; see Figure~\ref{fig:sigma}.
The recursion~$\sigma(x)$ has a straightforward combinatorial interpretation:
We consider the Dyck subpaths of the Dyck path~$x$ with increasing height levels and from left to right on each level, and we process them in two phases.
In phase~1, we flip the steps of each such subpath alternatingly between the rightmost and leftmost step, moving upwards.
In phase~2, we flip the steps alternatingly between the leftmost and rightmost step, moving downwards again.
We emphasize here that during the recursive computation of~$\sigma(x)$, no steps of~$x$ are ever flipped, but we always consider the same Dyck path and its Dyck subpaths as function arguments.

Note that by the definition \eqref{eq:sigma-rec}, we have $|\sigma(x)|=2|u|+2$, where $x=1u0v$ is the canonical decomposition.
We let~$P_\sigma(x)$ denote the sequence of vertices obtained by starting at the vertex~$x$ and flipping bits one after the other according to the sequence~$\sigma(x)$.
The following properties were proved in~\cite[Proposition~2]{mlc-short:18}.
\begin{enumerate}[label=(\roman*),topsep=0mm,leftmargin=7mm]
\item For any $x\in D_n$, $P_\sigma(x)$ is a path in the graph~$H_n$.
Moreover, all paths in $\cP_n:=\{P_\sigma(x) \mid x\in D_n\}$ are disjoint, and together they visit all vertices of~$H_n$.
\item For any first vertex $x\in D_n$, considering the canonical decomposition $x=1u0v$, the last vertex of~$P_\sigma(x)$ is given by $u01v\in D_n^-$.
Consequently, the sets of first and last vertices of the paths~$\cP_n$ are given by~$D_n$ and~$D_n^-$, respectively.
\end{enumerate}

Table~\ref{tab:paths-q6} shows the five paths in~$\cP_n$ for~$n=3$. 

\begin{table}
\caption{The five paths in~$\cP_n$ in the graph~$H_n$ for~$n=3$ obtained from the flip sequences~$\sigma(x)$, $x\in D_n$.
The gray boxes highlight the (possibly empty) substrings of~$x$ corresponding to the subpaths~$u$ or~$v$ in the canonical decomposition of~$x$.
}
\newcommand{\upart}[1]{\fcolorbox{white}{gray}{\hspace{-0.8mm}#1\hspace{-0.8mm}}}
\newcommand{\vpart}[1]{\fcolorbox{white}{lightgray}{\hspace{-0.8mm}#1\hspace{-0.8mm}}}
\begin{tabular}{|l|l|l|}
\hline
First vertex $x\in D_n$ & Flip sequence~$\sigma(x)$ & Last vertex $x\in D_n^-$ \\ \hline
1\upart{1100}0\vpart{\vphantom{1}} & $(6,1,5,2,4,3,2,4,1,5)$ & \upart{1010}01\vpart{\vphantom{1}} \\
1\upart{1010}0\vpart{\vphantom{1}} & $(6,1,3,2,1,3,5,4,3,5)$ & \upart{1100}01\vpart{\vphantom{1}} \\
1\upart{10}0\vpart{10} & $(4,1,3,2,1,3)$ & \upart{10}01\vpart{10} \\
1\upart{\vphantom{1}}0\vpart{1100} & $(2,1)$ & 0\upart{\vphantom{1}}1\vpart{1100} \\
1\upart{\vphantom{1}}0\vpart{1010} & $(2,1)$ & 0\upart{\vphantom{1}}1\vpart{1010} \\ \hline
\end{tabular}
\label{tab:paths-q6}
\end{table}

\subsection{Flippable pairs}
\label{sec:flippable-pairs}

To compute a Hamilton cycle in the middle levels graphs~$G_n$, we apply small local modifications to certain pairs of paths from~$\cP_n$, giving us additional freedom in combining an appropriate set of paths to a Hamilton cycle.
Specifically, we say that $x,y \in D_n$ form a \emph{flippable pair}~$(x,y)$, if
\begin{equation}
\label{eq:xy}
\begin{aligned}
  x &= 110w0v \enspace, \\
  y &= 101w0v
\end{aligned}
\end{equation}
for some~$w,v\in D$.
In terms of rooted trees, the tree~$y$ is obtained from~$x$ by removing the pending edge that leads to the leftmost child of the leftmost child of the root, and by attaching this edge leftmost to the root; see Figure~\ref{fig:xypair} (recall the correspondence between Dyck paths and rooted trees explained in Section~\ref{sec:basic-defs}).
We denote this operation by~$\tau$ and write~$y=\tau(x)$.
The preimage $T_n\seq D_n$ of the mapping~$\tau$ are all rooted trees with $n$~edges of the form $x=110w0v$, and the image of~$\tau$ are all rooted trees of the form $y=101w0v$, where $w,v\in D$; see Figure~\ref{fig:rot}.
Note that these two sets are disjoint.

\begin{figure}
\includegraphics[scale=0.916]{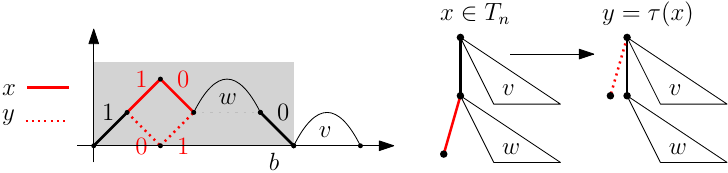}
\caption{A~flippable pair~$(x,y)$ and its Dyck path interpretation (left) and rooted tree interpretation (right).}
\label{fig:xypair}
\end{figure}

\begin{figure}
\includegraphics[scale=0.916]{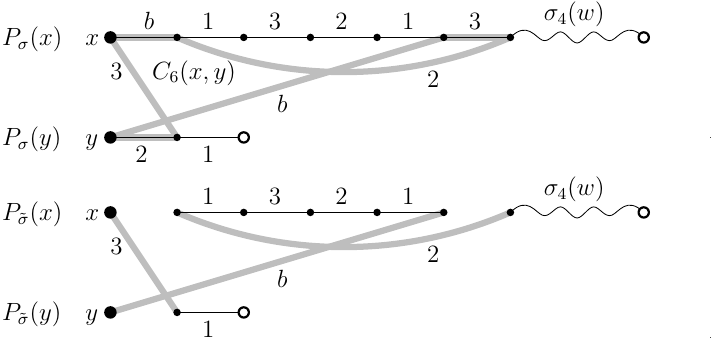}
\caption{
The paths~$P_\sigma(x)$ and~$P_\sigma(y)$ (black) for a flippable pair~$(x,y)$ traverse a common 6-cycle~$C_6(x,y)$ (gray).
The symmetric difference yields paths~$P_\tsigma(x)$ and~$P_\tsigma(y)$ that have interchanged end vertices.
The numbers on the edges indicate the flipped bit positions.
}
\label{fig:c6xy}
\end{figure}

Evaluating the recursion \eqref{eq:sigma-rec} for the bitstrings in a flippable pair~$(x,y)$ as in \eqref{eq:xy} yields
\begin{align*}
 \sigma(x) &= (b,1,3,2,1,3,\sigma_4(w)) \enspace, \\
 \sigma(y) &= (2,1) \enspace,
\end{align*}
where~$b:=|w|+2$.
It follows that the paths~$P_\sigma(x)$ and~$P_\sigma(y)$ intersect a common 6-cycle~$C_6(x,y)$ in the graph~$H_n$ as shown in Figure~\ref{fig:c6xy}.
Specifically, this 6-cycle can be encoded by
\begin{equation}
\label{eq:c6xy}
  C_6(x,y)=1{*}{*}w{*}v \enspace,
\end{equation}
where the six cycle vertices are obtained by substituting the three~$*$s by all six combinations of symbols from~$\{0,1\}$ that use each symbol at least once.
Consequently, taking the symmetric difference of the edge sets of~$P_\sigma(x)$ and~$P_\sigma(y)$ with the 6-cycle~$C_6(x,y)$ yields two paths on the same vertex set as~$P_\sigma(x)$ and~$P_\sigma(y)$, but with interchanged end vertices.
The resulting paths~$P_\tsigma(x)$ and~$P_\tsigma(y)$ have flip sequences
\begin{equation}
\label{eq:tsigma}
\begin{aligned}
  \tsigma(x) &:= (3,1) \enspace, \\
  \tsigma(y) &:= (b,1,2,3,1,2,\sigma_4(w)) \enspace,
\end{aligned}
\end{equation}
and we refer to these two paths as \emph{flipped paths} corresponding to the flippable pair~$(x,y)$.

\subsection{The Hamilton cycle algorithm}
\label{sec:hamcycle}

In this section we present our algorithm to compute a Hamilton cycle in the middle levels graph~$G_n$ (Algorithm~\ref{alg:hamcycle}).
The Hamilton cycle is obtained by combining paths that are computed via the flip sequences~$\sigma$ and~$\tsigma$.
We use the decomposition of~$G_n$ into~$H_n\, 0$, $\ol{\rev}(H_n)\, 1$, plus the matching~$M_n$ discussed in Section~\ref{sec:basic-defs}; see Figure~\ref{fig:2factor}.
By property~(ii) from Section~\ref{sec:paths}, the sets of first and last vertices of the paths~$\cP_n$ are~$D_n$ and~$D_n^-$, respectively.
It is easy to see that these two sets are preserved under the mapping~$\ol{\rev}$.
Together with property~(i) from Section~\ref{sec:paths} it follows that
\begin{equation}
\label{eq:2-factor}
  \cC_n:=\cP_n\, 0 \;\cup\; \ol{\rev}(\cP_n)\, 1 \;\cup\; M_n'
\end{equation}
with $M_n':=\{(x0,x1)\mid x\in D_n\cup D_n^-\}\seq M_n$ is a 2-factor in the middle levels graph.
A~\emph{2-factor} in a graph is a collection of disjoint cycles that together visit all vertices of the graph.
Note that along each of the cycles in the 2-factor, the paths from~$\cP_n\, 0$ are traversed in forward direction, and the paths from~$\ol{\rev}(\cP_n)\, 1$ in backward direction.
Observe also that the definition of flippable pairs given in Section~\ref{sec:flippable-pairs} allows us to replace in the definition \eqref{eq:2-factor} any two paths~$P_\sigma(x)$ and~$P_\sigma(y)$ from~$\cP_n$ for which~$(x,y)$ forms a flippable pair by the corresponding flipped paths~$P_\tsigma(x)$ and~$P_\tsigma(y)$, yielding another 2-factor.
Specifically, if the paths we replace lie on two different cycles, then the replacement will join the two cycles to one cycle.
The final algorithm makes all those choices such that the resulting 2-factor consists only of a single cycle, i.e., a Hamilton cycle.

\begin{figure}
\centering
\PSforPDF{
 \psfrag{cc}{$\cC_n$}
 \psfrag{gn}{$G_n$}
 \psfrag{hn}{$H_n\, 0$}
 \psfrag{hnrev}{$\ol{\rev}(H_n)\, 1$}
 \psfrag{b2n0}{$B_{2n,n}\, 0$}
 \psfrag{b2n1}{$B_{2n,n}\, 1$}
 \psfrag{mn}{$M_n$}
 \psfrag{pn}{$\cP_n$}
 \psfrag{pnrev}{$\ol{\rev}(\cP_n)$}
 \includegraphics{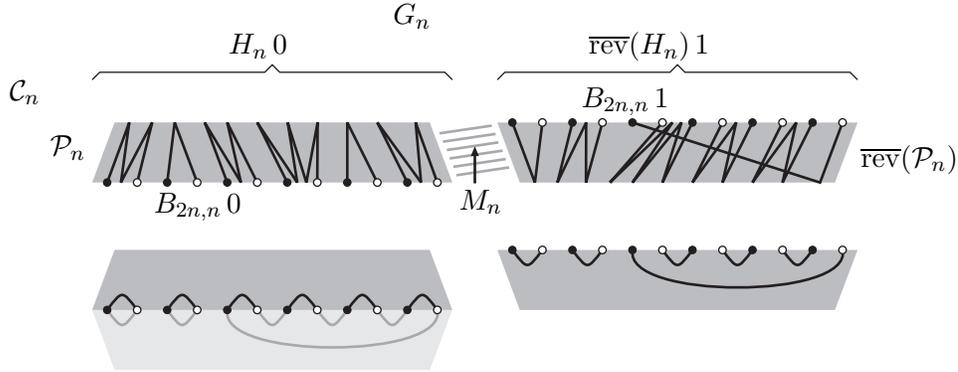}
}
\caption{The top part shows the decomposition of the middle levels graph~$G_n$ and the definition~\eqref{eq:2-factor}.
In this example, the 2-factor~$\cC_n$ consists of three disjoint cycles that together visit all vertices of the graph.
The sets~$D_n$ and~$D_n^-$ of first and last vertices of the paths~$\cP_n$ and~$\ol{\rev}(\cP_n)$ are drawn in black and white, respectively.
The bottom part shows a simplified drawing that depicts the structure of the 2-factor, which consists of two short cycles and one long cycle.}
\label{fig:2factor}
\end{figure}

\begin{algorithm}
\renewcommand\theAlgoLine{H\arabic{AlgoLine}}
\LinesNumbered
\DontPrintSemicolon
\SetEndCharOfAlgoLine{}
\SetNlSty{}{}{}
\SetArgSty{}
\SetKw{KwIf}{if}
\SetKw{KwElse}{else}
\SetKw{KwThen}{then}
\SetKw{KwDownto}{downto}
\caption[Algorithm HamCycle]{\mbox{$\HamCycle(n,x,\ell)$}}
\label{alg:hamcycle}
\vspace{.2em}
\KwIn{An integer~$n\geq 1$, a vertex $x\in G_n$, an integer~$\ell\geq 1$}
\KwResult{Starting from the vertex~$x$, the algorithm visits the next $\ell$~vertices on a Hamilton cycle in~$G_n$}
\vspace{.2em}
$(y,i):=\Init(n,x,\ell)$ \label{line:call-init} \;
\While{$\true$} {
  $y^-:=y_1 y_2\ldots y_{2n}$ \tcc*[r]{ignore last bit of~$y$} \label{line:ym}
  \KwIf $\big((y^-\in T_n \text{ and } \IsFlipTree(y^-)) \text{ or } (y^-\in \tau(T_n) \text{ and } \IsFlipTree(\tau^{-1}(y^-))\big)$ \KwThen $s:=\tsigma(y^-)$ \tcc*[r]{compute flip sequence~$\tsigma$ ...} \label{line:call-tsigma}
  \KwElse $s:=\sigma(y^-)$ \tcc*[r]{... or~$\sigma$} \label{line:call-sigma1}
  \For(\tcc*[f]{flip bits according to sequence~$s$}){$j:=1$ \KwTo $|s|$ \label{line:forward-path}} {
    $y_{s_j} := 1-y_{s_j}$ \tcc*[r]{flip bit at position~$s_j$} \label{line:flip-forward}
    $\Visit(y)$ \label{line:visit1} \;
    \KwIf $(i:=i+1)=\ell$ \Return \label{line:done1} \;
  }
  $y_{2n+1}:=1$ \tcc*[r]{flip last bit to 1} \label{line:jump-to-1copy}
  $\Visit(y)$ \label{line:visit2} \;
  \KwIf $(i:=i+1)=\ell$ \Return \label{line:done2} \;
  $u01v:=y_1 y_2 \ldots y_{2n}$ \tcc*{canonical decomposition of $y_1 \ldots y_{2n}\in D_n^-$} \label{line:decomposition}
  $s:=\sigma(1 \,\ol{\rev}(v)\, 0 \,\ol{\rev}(u))$ \tcc*[r]{compute flip sequence~$\sigma$} \label{line:call-sigma2}
  \For(\tcc*[f]{flip bits according to reverse sequence~$s$}){$j:=|s|$ \KwDownto $1$ \label{line:backward-path}} {
    $y_{2n+1-s_j} := 1-y_{2n+1-s_j}$ \tcc*[r]{flip bit at position $2n+1-s_j$} \label{line:flip-backward}
    $\Visit(y)$ \label{line:visit3} \;
    \KwIf $(i:=i+1)=\ell$ \Return \label{line:done3} \;
  }
  $y_{2n+1}:=0$ \tcc*[r]{flip last bit to 0} \label{line:jump-to-0copy}
  $\Visit(y)$ \label{line:visit4} \;
  \KwIf $(i:=i+1)=\ell$ \Return \label{line:done4} \;
}
\end{algorithm}

For the rest of the paper we will focus on proving that the algorithm $\HamCycle$, described in Algorithm~\ref{alg:hamcycle}, implies Theorem~\ref{thm:algo2}.
$\HamCycle$ is called with three input parameters: $n$~determines the length~$2n+1$ of the bitstrings, $x$~is the starting vertex of the Hamilton cycle and must have weight~$n$ or~$n+1$, and $\ell$~is the number of vertices to visit along the cycle.

The variable~$y$ is the current vertex along the cycle, and the variable~$i$ counts the number of vertices that have already been visited.
The calls $\Visit(y)$ in lines~\ref{line:visit1}, \ref{line:visit2}, \ref{line:visit3} and \ref{line:visit4} indicate where a function using our Hamilton cycle algorithm could perform further operations on the current vertex~$y$.
Each time a vertex along the cycle is visited, we increment~$i$ and check whether the desired number~$\ell$ of vertices has been visited (lines~\ref{line:done1}, \ref{line:done2}, \ref{line:done3} and \ref{line:done4}).

We postpone the definition of the functions $\Init$ and $\IsFlipTree$ called in lines \ref{line:call-init} and \ref{line:call-tsigma} a little bit, and assume for a moment that the input vertex~$x$ of the middle levels graph~$G_n$ has the form $x=z0$ with~$z\in D_n$.
In this case the variables~$y$ and~$i$ will be initialized to~$y:=x$ and~$i:=1$ in line~\ref{line:call-init}.
Let us also assume that the return value of $\IsFlipTree$ called in line~\ref{line:call-tsigma} is always~$\false$.
With these simplifications, the algorithm $\HamCycle$ computes exactly the 2-factor~$\cC_n$ defined in \eqref{eq:2-factor} in the middle levels graph~$G_n$.

Indeed, one complete execution of the first for-loop corresponds to following one path from the set~$\cP_n\, 0$ in the graph~$H_n\, 0$ starting at its first vertex and ending at its last vertex, and one complete execution of the second for-loop corresponds to following one path from the set~$\ol{\rev}(\cP_n)\, 1$ in the graph~$\ol{\rev}(H_n)\, 1$ starting at its last vertex and ending at its first vertex.
At the intermediate steps in lines~\ref{line:jump-to-1copy} and \ref{line:jump-to-0copy}, the last bit is flipped.
These flips correspond to traversing an edge from the matching~$M_n'$.
The paths~$\cP_n$ are computed in lines~\ref{line:call-sigma1} and \ref{line:call-sigma2} using the recursion~$\sigma$, and the resulting flip sequences are applied in the two inner for-loops (line~\ref{line:flip-forward} and \ref{line:flip-backward}).
Note that if a path from the set~$\cP_n$ has $y\in D_n^-$ as a last vertex and if $y=u01v$ is the canonical decomposition, then $\ol{\rev}$~maps the last vertex of the path~$P\in \cP_n$ that has $1 \,\ol{\rev}(v)\, 0 \,\ol{\rev}(u)$ as first vertex onto~$y$.
This is a consequence of property~(ii) from Section~\ref{sec:paths}, from which we obtain that the last vertex of~$P$ is $\ol{\rev}(v)\, 0 1 \,\ol{\rev}(u)$, and applying~$\ol{\rev}$ to this vertex indeed yields~$y$.
From these observations and the definitions in lines~\ref{line:decomposition}, \ref{line:call-sigma2} and \ref{line:flip-backward} it follows that the paths in the second set on the right hand side of \eqref{eq:2-factor} are indeed traversed in backward direction (starting at the last vertex and ending at the first vertex).

We now explain the significance of the function $\IsFlipTree$ called in line~\ref{line:call-tsigma} of our algorithm.
This function interprets the current first vertex $y^-\in T_n\seq D_n$ or $\tau^{-1}(y^-)\in T_n$ as a rooted tree, and whenever it returns~$\true$, then instead of computing the flip sequence~$\sigma(y^-)$ in line~\ref{line:call-sigma1}, the algorithm computes the modified flip sequence~$\tsigma(y^-)$ in line~\ref{line:call-tsigma}.
Consequently, the function $\IsFlipTree$ controls which pairs of paths from~$\cP_n$, whose first vertices form a flippable pair, are replaced by the corresponding flipped paths, so that the resulting 2-factor in the middle levels graph~$G_n$ is a Hamilton cycle.
Observe that these modifications only apply to the set~$\cP_n\, 0$, but not to the set~$\ol{\rev}(\cP_n)\, 1$ on the right hand side of~\eqref{eq:2-factor}.

The function $\IsFlipTree$ therefore encapsulates the core `intelligence' of our Hamilton cycle algorithm.
We define this function and the function $\Init$ in the next two sections.
The correctness proof for the algorithm $\HamCycle$ is provided in Section~\ref{sec:correctness} below.

\subsection{The function \texorpdfstring{$\IsFlipTree$}{IsFlipTree}}
\label{sec:fliptree}

To define the Boolean function $\IsFlipTree$, we need a few more definitions related to trees.

\textit{Leaves, stars, and tree center.}
Any vertex of degree~1 of a tree is called a \emph{leaf}.
We call a leaf \emph{thin}, if its unique neighbor in the tree has degree~2.
A~\emph{star} is a tree in which all but at most one vertex are leaves.
The \emph{center} of a tree is the set of vertices that minimize the maximum distance to any other vertex.
Any tree either has a unique center vertex, or two center vertices that are adjacent.
For a rooted tree, these notions are independent of the vertex orderings.
Also note that the root of a rooted tree can be a leaf.

The following auxiliary function $\troot$ computes a canonically rooted version of a given rooted tree.
Formally, for any tree $x\in D_n$ and any integer~$r\geq 0$ the return value of $\troot(\trot^r(x))$ is the same rotated version of~$x$.
In the following functions, all comparisons between trees are performed using the bitstring representation.

\textit{The function $\troot$:}
Given a tree $x\in D_n$, first compute its center vertex/vertices.
If there are two centers~$c_1$ and~$c_2$, then let~$x'$ be the tree obtained by rooting~$x$ so that $c_1$~is the root and $c_2$~its leftmost child, let~$x''$ be the tree obtained by rooting~$x$ so that $c_2$~is the root and $c_1$~its leftmost child, and return the tree from~$\{x',x''\}$ with the lexicographically smaller bitstring representation.
If the center~$c$ is unique, then let $y_1,y_2,\ldots,y_k$ be the subtrees of~$x$ rooted at~$c$.
Consider the bitstring representations ot these subtrees, and compute the lexicographically smallest rotation of the string $(-1,y_1,-1,y_2,-1,\ldots,-1,y_k)$ using Booth's algorithm~\cite{MR585391}.
Here $-1$~is an additional symbol that is lexicographically smaller than~0 and~1, ensuring that the lexicographically smallest string rotation starts at a tree boundary.
Let~$\xhat$ be the tree obtained by rooting~$x$ at~$c$ such that the subtrees $y_1,\ldots,y_k$ appear exactly in this lexicographically smallest ordering, and return~$\xhat$.

We are now ready to define the function $\IsFlipTree$.

\textit{The function $\IsFlipTree$:}
Given a tree $x\in T_n$, return~$\false$ if $x$~is a star.
Otherwise compute~$\xhat:=\troot(x)$.
If $x$~has a thin leaf, then let~$x'$ be the tree obtained by rotating~$\xhat$ until it has the form $x'=1100 v$ for some~$v\in D$.
Return~$\true$ if~$x=x'$, and return~$\false$ otherwise.
If $x$~has no thin leaf, then let~$x'$ be the tree obtained by rotating~$\xhat$ until it has the form $x'=1(10)^k 0 v$ for some~$k\geq 2$ and~$v\in D$.
Return~$\true$ if~$x=x'$ and if the condition $v=(10)^l$ implies that~$l\geq k$, and return~$\false$ otherwise.

\subsection{The function \texorpdfstring{$\Init$}{Init}}
\label{sec:init}

It remains to define the function $\Init(n,x,\ell)$ called in line~\ref{line:call-init} of the algorithm $\HamCycle$.
This function moves forward along the Hamilton cycle from the given vertex~$x$ in~$G_n$ and visits all vertices until the first vertex of the form~$z0$ with $z\in D_n$ is encountered.
We then initialize the current vertex as~$y:=z0$, and set the vertex counter~$i$ to the number of vertices visited along the cycle from~$x$ to~$y$.
The parameter $\ell$ is passed, as it might be so small that the termination condition is already reached on this initial path.

The first task is to compute, for the given vertex~$x$, which path~$P_\sigma(y)\, 0$ or~$\ol{\rev}(P_\sigma(y))\, 1$ the vertex~$x$ is contained in.
With this information we can run essentially one iteration of the while-loop of the algorithm $\HamCycle$, after which we reach the first vertex of the form~$z0$ with~$z\in D_n$.

To achieve this, if the last bit of~$x$ is 0, i.e., $x=z0$, then we compute $y\in D_n$ such that $z\in P_\sigma(y)$ as follows:
We consider the point(s) with lowest height on the lattice path~$z$.
If the lowest point of~$z$ is unique, then we partition~$z$ uniquely as
\begin{subequations}
\label{eq:zpartition}
\begin{equation}
\begin{aligned}
z = \begin{cases}
     u_1\, 0 \,u_2\, 0 \ldots u_d\, 0 \,w\, 0 1 \,v_d\, 1 \,v_{d-1}\, 1 \ldots v_1\, 1 \,v & \text{if } z\in B_{2n,n} \enspace, \\
     u_1\, 0 \,u_2\, 0 \ldots u_d\, 0 1 \,w\, 1 \,v_d\, 1 \,v_{d-1}\, 1 \ldots v_1\, 1 \,v & \text{if } z\in B_{2n,n+1}
     \end{cases}
\end{aligned}
\end{equation}
for some~$d\geq 0$ and $u_1,\ldots,u_d,v_1,\ldots,v_d,w\in D$.
If there are at least two lowest points of~$z$, then we partition~$z$ uniquely as
\begin{equation}
\begin{aligned}
z = \begin{cases}
    u_1\, 0 \,u_2\, 0 \ldots u_d\, 0 1 \,w\, 0 \,v_d\, 1 \,v_{d-1}\, 1 \ldots v_1\, 1 \,v & \text{if } z\in B_{2n,n} \enspace, \\
    u_1\, 0 \,u_2\, 0 \ldots u_d\, 1 \,w\, 0 1 \,v_d\, 1 \,v_{d-1}\, 1 \ldots v_1\, 1 \,v & \text{if } z\in B_{2n,n+1}
     \end{cases}
\end{aligned}
\end{equation}
\end{subequations}
for some~$d\geq 0$ and $u_1,\ldots,u_d,v_1,\ldots,v_d,w\in D$.
In all four cases, a straightforward calculation using the definition \eqref{eq:sigma-rec} shows that
\begin{equation*}
  y:=1 \,u_1\, 1 \,u_2\ldots 1 \,u_d\, 1 \,w\, 0 \,v_d\, 0 \,v_{d-1}\, 0 \ldots v_1\, 0 \,v
\end{equation*}
is indeed the first vertex of the path~$P_\sigma(y)$ that contains the vertex~$z$.
In particular, if $z\in D_n$, then~$d=0$ and $y=z=1 w 0 v$.

If the last bit of~$x$ is 1, i.e., $x=z1$, then we compute $y\in D_n$ such that $z\in \ol{\rev}(P_\sigma(y))$ by applying the previous steps to the vertex~$\ol{\rev}(z)$.

For more details how the function $\Init$ is implemented, see our C++ implementation~\cite{cos}.

\section{Correctness of the algorithm}
\label{sec:correctness}

The properties (i) and (ii) of the paths~$\cP_n$ claimed in Section~\ref{sec:paths} were proved in~\cite[Proposition~2]{mlc-short:18}.
Consequently, if we assume for a moment that the function $\IsFlipTree$ always returns~$\false$, then the arguments given in Section~\ref{sec:hamcycle} show that the algorithm $\HamCycle$ correctly follows one cycle of the 2-factor~$\cC_n$ defined in \eqref{eq:2-factor} in the middle levels graph~$G_n$.
Moreover, by the symmetric definition in line~\ref{line:call-tsigma}, for any flippable pair~$(x,y)$ either the flip sequence~$\sigma$ is applied to both~$x$ and~$y$, or the modified flip sequence~$\tsigma$ is applied to both~$x$ and~$y$.
Consequently, by the definition of flippable pairs given in Section~\ref{sec:flippable-pairs}, the algorithm $\HamCycle$ correctly computes a 2-factor in the graph~$G_n$ for \emph{any} Boolean function $\IsFlipTree$ on the set~$T_n$.
It remains to argue that for the particular Boolean function $\IsFlipTree$ defined in the previous section, our algorithm indeed computes a Hamilton cycle.
For this we analyze the structure of the 2-factor~$\cC_n$, which is best described by yet another family of trees.

\begin{figure}
\includegraphics[scale=0.916]{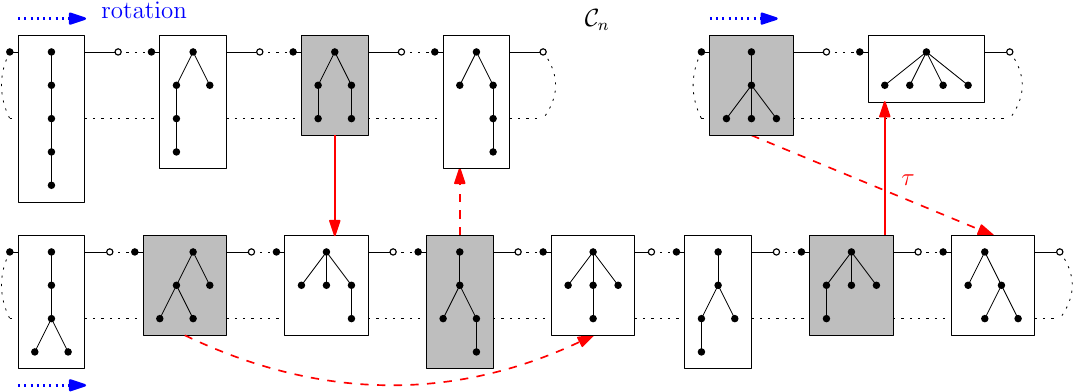}
\caption{Tree rotation along the cycles of the 2-factor~$\cC_n$ (dotted lines), the mapping~$\tau$ (solid and dashed arrows) and the graph~$\cT_n$ for~$n=4$.
The highlighted trees form the set~$T_n$, the preimage of~$\tau$.
The graph~$\cT_n$ has the three cycles of~$\cC_n$ as nodes, and it contains only the two solid directed edges (not the dashed ones), forming a spanning tree on the cycles of~$\cC_n$.
}
\label{fig:rot}
\end{figure}

\textit{Plane trees.}
A~\emph{plane tree} is a tree with a specified cyclic ordering of the neighbors of each vertex.
We think of a plane tree as a tree embedded in a plane, where the neighbors of each vertex appear exactly in the specified ordering in counterclockwise direction around that vertex; see Figure~\ref{fig:t6}.
Equivalently, plane trees arise as equivalence classes of rooted trees under rotation.

It was shown in~\cite[Proposition~2]{mlc-short:18} that for any cycle from~$\cC_n$, if we consider two consecutive vertices of the form~$x0$ and~$y0$ with $x,y\in D_n$ on the cycle and the canonical decomposition of $x=1u0v$, then we have $y=u1v0$.
In terms of rooted trees, we have~$y=\trot(x)$.
Consequently, the set of cycles of~$\cC_n$ is in bijection with the equivalence classes of all rooted trees with $n$~edges under rotation; see Figure~\ref{fig:rot}.
In particular, the number of cycles in the 2-factor is given by the number of plane trees with $n$~edges.

The definition of flippable pairs~$(x,y)$ given in Section~\ref{sec:flippable-pairs} shows that if~$P_\sigma(x)$ and~$P_\sigma(y)$ are contained in two distinct cycles of~$\cC_n$, then replacing these two paths by the flipped paths~$P_\tsigma(x)$ and~$P_\tsigma(y)$ joins the two cycles to a single cycle on the same set of vertices; recall Figure~\ref{fig:c6xy}.
As mentioned before, this replacement operation is equivalent to taking the symmetric difference of the edge sets of~$P_\sigma(x)$ and~$P_\sigma(y)$ with the 6-cycle~$C_6(x,y)$ defined in \eqref{eq:c6xy}.
The following two properties were established in~\cite[Proposition~3]{mlc-short:18}:
For any flippable pairs~$(x,y)$ and~$(x',y')$, the 6-cycles~$C_6(x,y)$ and~$C_6(x',y')$ are edge-disjoint.
Moreover, for any flippable pairs~$(x,y)$ and~$(x,y')$, the two pairs of edges that the two 6-cycles~$C_6(x,y)$ and~$C_6(x,y')$ have in common with the path~$P_\sigma(x)$ are not interleaved, but one pair appears before the other pair along the path.
Consequently, none of the 6-cycles used for these joining operations interfere with each other.

Using these observations, we define an auxiliary graph~$\cT_n$ as follows.
The nodes of~$\cT_n$ are the equivalence classes of rooted trees with $n$~edges under rotation, which can be interpreted as plane trees.
For any flippable pair~$(x,y)$ for which $\IsFlipTree(x)=\true$, we add a directed edge from the equivalence class containing~$x$ to the equivalence class containing~$y=\tau(x)$ to the graph~$\cT_n$.
The graph~$\cT_n$ is shown in Figure~\ref{fig:rot} and Figure~\ref{fig:t6} for~$n=4$ and~$n=6$, respectively.
By what we said before, the nodes of~$\cT_n$ correspond to the cycles of the 2-factor~$\cC_n$, and the edges correspond to the flipped pairs of paths used by the algorithm $\HamCycle$.
To complete the correctness proof for our algorithm, it thus remains to prove the following lemma.

\begin{figure}
\includegraphics[scale=0.916]{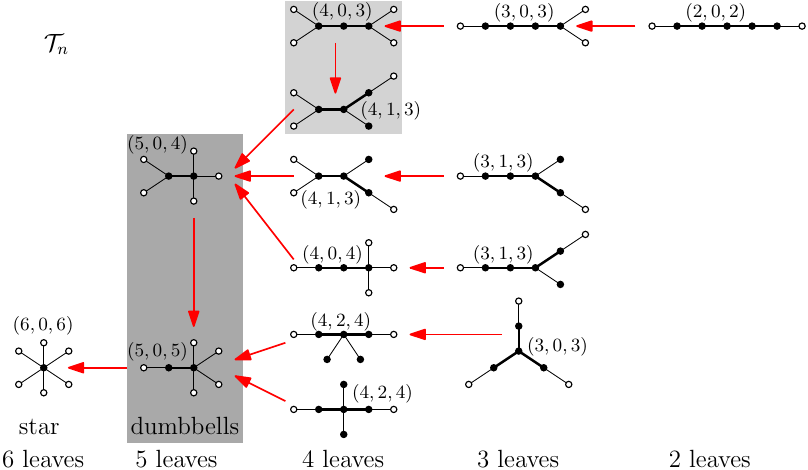}
\caption{The graph~$\cT_n$ for~$n=6$, where only the plane trees corresponding to each node are shown.
The skeleton of each plane tree is drawn with bold edges.
Terminal leaves are drawn white, and non-terminal leaves are drawn black.
The triples assigned to each plane tree are the signatures (number of leaves, number of non-terminal leaves, maximum degree).
The light-gray box highlights an edge of~$\cT_n$ along which the number of leaves stays the same, but the number of non-terminal leaves goes up from~0 to~1.
The dark-gray box highlights all dumbbells and an edge of~$\cT_n$ along which the number of leaves and non-terminal leaves stays the same, but the maximum degree goes up from~4 to~5.
}
\label{fig:t6}
\end{figure}

\begin{lemma}
For any~$n\geq 1$, the graph~$\cT_n$ is a spanning tree.
\end{lemma}

\begin{proof}
For the reader's convenience, the following definitions are illustrated in Figure~\ref{fig:t6}.
All these notions apply to rooted trees and to plane trees.
The \emph{skeleton} of a tree is the tree obtained by removing all leaves.
A~leaf of a tree is called \emph{terminal}, if it is adjacent to a leaf of the skeleton.
A~\emph{dumbbell} is a tree in which all but exactly two vertices are leaves.
Equivalently, a dumbbell has a skeleton consisting of a single edge, or a dumbbell is a tree with $n$~edges and $n-1$~leaves.
Note that any tree that is not a star has a skeleton with at least one edge, and any tree that is neither a star nor a dumbbell has a skeleton with at least two edges.
Also note that any thin leaf is a terminal leaf, but not every terminal leaf is thin (consider a dumbbell).

Consider a directed edge in the graph~$\cT_n$ that arises from two rooted trees~$(x,y)$ with~$y=\tau(x)$ and $\IsFlipTree(x)=\true$.
By the definition of the function $\IsFlipTree$, $x$~is not a star, so in particular we have~$n\geq 3$.
Moreover, exactly one of the following two conditions holds.
Case~(a): $x=1100v$ for some~$v\in D$.
In this case, we have $y=\tau(x)=1010v$, and as~$v$ is nonempty by the condition~$n\geq 3$, we obtain that $y$~has one more leaf than~$x$.
Case~(b): $x=1(10)^k0v$ for some~$k\geq 2$ and~$v\in D$, and $v=(10)^l$ implies that~$l\geq k$.
In this case, as $x$~is not a star, the subtree~$v$ has at least one edge, so the root of~$x$ is not a leaf.
Moreover, all leaves of~$x$ in the leftmost subtree are terminal leaves.
We distinguish two subcases.
Case~(b1): $v$ is not a star rooted at the center, i.e., $x$~is not a dumbbell.
In this case, one of the terminal leaves in the leftmost subtree of~$x$ becomes a non-terminal leaf in $y=\tau(x)=101(10)^{k-1}0v$, so~$y$ and~$x$ have the same number of leaves, but~$y$ has one more non-terminal leaf.
Case~(b2): $v$ is a star rooted at the center, i.e., $x$~is a dumbbell and $v=(10)^l$ for some~$l\geq k$.
In this case, the root of~$x$ is a vertex that has maximum degree~$l+1$, whereas the root of the dumbbell $y=\tau(x)=101(10)^{k-1}0(10)^l$ has degree~$l+2$, so both dumbbells~$y$ and~$x$ have the same number of leaves and non-terminal leaves, but the maximum degree of~$y$ is one higher than that of~$x$.

To every plane tree~$T$ we therefore assign a \emph{signature} $s(T):=(\ell,t,d)$, where~$\ell$ is the number of leaves of~$T$, $t$~is the number of non-terminal leaves of~$T$, and $d$~is the maximum degree of~$T$; see Figure~\ref{fig:t6}.
By our observations from before, for any directed edge~$(T,T')$ between two plane trees~$T$ and~$T'$ in~$\cT_n$, when comparing $s(T)=:(\ell,t,d)$ and $s(T')=:(\ell',t',d')$, then either $\ell<\ell'$ in case~(a) from before, or $\ell=\ell'$ and $t<t'$ in case~(b1) from before, or $(\ell,t)=(\ell',t')$ and $d<d'$ in case~(b2) from before.
It follows that $\cT_n$~does not have any cycles where all edges have the same orientation.
In particular, $\cT_n$~has no loops.

Moreover, as a consequence of the initial canonical rooting performed by the call to $\troot$, the function $\IsFlipTree$ returns~$\true$ for at most one tree from each equivalence class of rooted trees under rotation.
This implies that each node of~$\cT_n$ has out-degree at most~1.
Consequently, $\cT_n$~does not have a cycle where the edges have different orientations, as such a cycle would have a node with out-degree~2.
Combining these observations shows that $\cT_n$~is acyclic.

It remains to prove that $\cT_n$~is connected.
For this we show how to move from any plane tree~$T$ along the edges of~$\cT_n$ to the star with $n$~edges; see Figure~\ref{fig:t6}.
We assume that $T$~is not the star with $n$~edges, so in particular~$n\geq 3$.
If $T$~has a thin leaf, then we can clearly root~$T$ so that the rooted tree has the form $1100v$ for some~$v\in D$.
Consequently, there exists an edge in~$\cT_n$ that leads from~$T$ to a tree~$T'$ that has one more leaf than~$T$.
If $T$~has no thin leaf, then we can root~$T$ so that the rooted tree has the form $1(10)^k0v$ for some~$k\geq 2$ and~$v\in D$, and $v=(10)^l$ implies that~$l\geq k$.
To see this we distinguish two cases.
If $T$~is not a dumbbell, then the skeleton of~$T$ has at least two edges, so rooting~$T$ at any vertex in distance~1 from a leaf of the skeleton yields a rooted tree of the form $1(10)^k0v$ for some~$k\geq 2$ and~$v\in D$ where $v$~is not a star rooted at the center.
On the other hand, if $T$~is a dumbbell, then we can root~$T$ at a vertex of maximum degree so that the rooted tree has the form $1(10)^kv$ for some~$k\geq 2$ and $v=(10)^l$ with~$l\geq k$.
Consequently, there exists an edge in~$\cT_n$ that leads from~$T$ to a tree~$T'$ that has the same number of leaves, and either one more non-terminal leaf, or the same number of non-terminal leaves, but a maximum degree that is one higher than that of~$T$.
We repeat this argument, following directed edges of~$\cT_n$, until we arrive at the star with $n$~edges.

We have shown that $\cT_n$~is acyclic and connected, so it is indeed a spanning tree.
\end{proof}

\section{Running time and space requirements}
\label{sec:running-time}

\subsection{Running time}

For any $x\in D_n$, the flip sequence~$\sigma(x)$ can be computed in linear time.
To achieve this, we precompute an array of bidirectional pointers below the Dyck subpaths of~$x$ between corresponding pairs of an $\upstep$-step and $\downstep$-step on the same height; see Figure~\ref{fig:pointers}.
Using these pointers, each canonical decomposition operation encountered in the recursion \eqref{eq:sigma-rec} can be performed in constant time, so that the overall running time of the recursion is~$\cO(n)$.
Clearly, the sequence~$\tsigma(x)$ can also be computed in time~$\cO(n)$ by modifying the sequence~$\sigma(x)$ as described in \eqref{eq:tsigma} in constantly many positions.
Obviously, the functions~$\ol{\rev}(u)$ and~$\ol{\rev}(v)$ called in line~\ref{line:call-sigma2} can also be computed in time~$\cO(n)$.

To compute the functions $\troot$ and $\IsFlipTree$, we first convert the given bitstring $x\in D_n$ to a tree in adjacency list representation, which can clearly be done in time~$\cO(n)$; recall the correspondence between Dyck paths and rooted trees from Figure~\ref{fig:bij}.
The adjacency list representation allows us to compute the center vertex/vertices in linear time by removing leaves in rounds until only a single vertex or a single edge is left (see~\cite{DBLP:books/daglib/0022194}).
Moreover, it allows us to perform each rotation operation $\trot$ in constant time, and a full tree rotation in time~$\cO(n)$.
Booth's algorithm to compute the lexicographically smallest string rotation also runs in linear time~\cite{MR585391}.
This gives the bound~$\cO(n)$ for the time spent in the functions $\troot$ and $\IsFlipTree$.

It was shown in~\cite[Proposition~2]{mlc-short:18} that the distance between any two consecutive vertices of the form~$x0$ and~$y0$ with $x,y\in D_n$ on a cycle of the 2-factor~$\cC_n$ is exactly~$4n+2$.
Moreover, replacing a path~$P_\sigma(x)$ in the first set on the right hand side of \eqref{eq:2-factor} by the path~$P_\tsigma(x)$ does not change this distance; recall Figure~\ref{fig:c6xy}.
It follows that in each iteration of the while-loop of our algorithm $\HamCycle$, exactly~$4n+2$ vertices are visited.
Combining this with the time bounds~$\cO(n)$ derived for the functions~$\sigma$, $\tsigma$, $\ol{\rev}$ and $\IsFlipTree$ that are called once or twice during each iteration of the while-loop, we conclude that the while-loop takes time~$\cO(\ell+n)$ to visit $\ell$~vertices of the Hamilton cycle.

\begin{figure}
\centering
\PSforPDF{
 \psfrag{x}{$x\in D_n$}
 \includegraphics{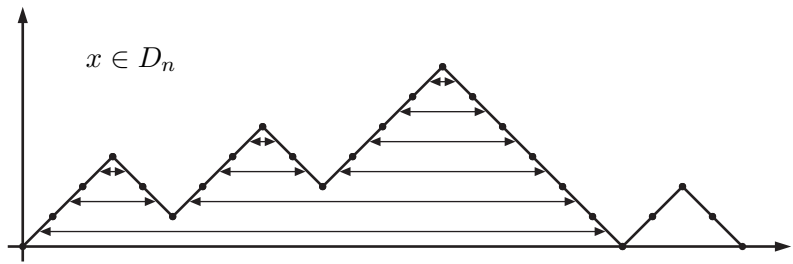}
}
\caption{Auxiliary pointers to compute the flip sequence~$\sigma(x)$ for any $x\in D_n$ in time~$\cO(n)$.
The lattice path~$x$ is the same as in Figure~\ref{fig:sigma}, and the resulting sequence~$\sigma(x)$ is shown in that figure.}
\label{fig:pointers}
\end{figure}

The function $\Init$ takes time~$\cO(n)$, as the partitions \eqref{eq:zpartition} can be computed in linear time, and as we visit at most linearly many vertices in this function (every path in~$\cP_n$ has only length~$\cO(n)$).

Combining the time bounds~$\cO(n)$ for the initialization phase and the time~$\cO(\ell+n)$ spent in the while-loop, we obtain the claimed overall bound~$\cO(\ell+n)$ for the algorithm $\HamCycle$.

\subsection{Space requirements}

Throughout our algorithm, we only store constantly many bitstrings of length~$2n$, rooted trees with $n$~edges, and flip sequences of length at most $4n+2=\cO(n)$, proving that the entire space needed is~$\cO(n)$.

\section*{Acknowledgements}

We thank the anonymous referee for the careful reading and many helpful comments that considerably improved the readability of this manuscript.

\bibliographystyle{alpha}
\bibliography{refs}

\end{document}